\documentclass[11pt]{article}

\usepackage{graphicx}
\usepackage{amsfonts}
\usepackage{amsmath}
\usepackage{amsthm}
\usepackage[mathcal]{euscript}      
\usepackage{bm}                     
\usepackage{color}
\usepackage{epstopdf}
\usepackage{amssymb}
\usepackage{mathrsfs,mathdots}
\usepackage{lscape}
\usepackage[version=3]{mhchem}

\usepackage[paper=a4paper,dvips,top=3cm,left=2.4cm,right=2.4cm,
    foot=1cm,bottom=3cm]{geometry}

\begin{document}

\title{Transposes, L-Eigenvalues and Invariants \\ of Third Order Tensors}
\author{Liqun Qi\footnote{%
    Department of Applied Mathematics, The Hong Kong Polytechnic University,
    Hung Hom, Kowloon, Hong Kong ({\tt maqilq@polyu.edu.hk}).
    This author's work was partially supported by the Hong Kong Research Grant Council
    (Grant No. PolyU  501913, 15302114, 15300715 and 15301716).}
    }

\date{\today}
\maketitle

\begin{abstract}  Third order tensors have wide applications in mechanics, physics and engineering.  The most famous and useful third order tensor is the piezoelectric tensor, which plays a key role in the piezoelectric effect, first discovered by Curie brothers.   On the other hand, the Levi-Civita tensor is famous in tensor calculus.    In this paper, we study third order tensors and (third order) hypermatrices systematically, by regarding a third order tensor as a linear operator which transforms a second order tensor into a first order tensor, or a first order tensor into a second order tensor.   For a third order tensor, we define its transpose, kernel tensor and L-eigenvalues.  Here, ``L'' is named after Levi-Civita.  The transpose of a third order tensor is uniquely defined.  In particular, the transpose of the piezoelectric tensor is the inverse piezoelectric tensor (the electrostriction tensor).
The kernel tensor of a third order tensor is a second order positive semi-definite symmetric tensor, which is the product of that third order tensor and its transpose.      We define L-eigenvalues, singular values, C-eigenvalues and Z-eigenvalues for a third order tensor.   They are all invariants of that third order tensor.     For a third order partially symmetric tensor, we give its eigenvector decomposition.  The piezoelectric tensor, the inverse piezoelectric tensor, and third order symmetric tensors are third order partially symmetric tensors.
We make a conjecture a third order symmetric tensor has seven independent invariants.  We raise the questions on how may independent invariants a third order tensor, a third order partially symmetric tensor and a third order cyclically symmetric tensor may have respectively.

\end{abstract}

\newtheorem{Theorem}{Theorem}[section]
\newtheorem{Definition}[Theorem]{Definition}
\newtheorem{Lemma}[Theorem]{Lemma}
\newtheorem{Corollary}[Theorem]{Corollary}
\newtheorem{Proposition}[Theorem]{Proposition}
\newtheorem{Conjecture}[Theorem]{Conjecture}
\newtheorem{Question}[Theorem]{Question}

\renewcommand{\hat}[1]{\widehat{#1}}
\renewcommand{\tilde}[1]{\widetilde{#1}}
\renewcommand{\bar}[1]{\overline{#1}}
\newcommand{\REAL}{\mathbb{R}}
\newcommand{\COMPLEX}{\mathbb{C}}
\newcommand{\SPHERE}{\mathbb{S}^2}
\newcommand{\diff}{\,\mathrm{d}}
\newcommand{\st}{\mathrm{s.t.}}
\newcommand{\T}{\top}
\newcommand{\vt}[1]{{\bf #1}}
\newcommand{\x}{{\vt{x}}}
\newcommand{\y}{{\vt{y}}}
\newcommand{\z}{{\vt{z}}}
\newcommand{\w}{{\vt{w}}}
\newcommand{\e}{{\vt{e}}}
\newcommand{\g}{{\vt{g}}}
\newcommand{\0}{{\vt{0}}}
\newcommand{\Ten}{\mathcal{T}}
\newcommand{\HH}{\mathbb{H}}
\newcommand{\A}{\mathcal{A}}
\newcommand{\B}{\mathcal{B}}
\newcommand{\C}{\mathcal{C}}
\newcommand{\D}{\mathcal{D}}
\newcommand{\E}{\mathcal{E}}
\newcommand{\OOO}{\mathcal{O}}
\newcommand{\U}{{\bf{U}}}
\newcommand{\V}{{\bf{V}}}
\newcommand{\W}{{\bf{W}}}
\newcommand{\I}{{\bf{I}}}
\newcommand{\OO}{{\bf{O}}}
\newcommand{\RESULTANT}{\mathrm{Res}}

\newpage
\section{Introduction}

Third order tensors have wide applications in mechanics, physics and engineering.  The most famous and useful third order tensor is the piezoelectric tensor \cite{CJQ-17, CC-80, Ha-07, KPG-08, Lo-89, ZTP-13}, which plays a key role in the piezoelectric effect, first discovered by Curie brothers \cite{CC-80}.   Other third order tensors in physics and engineering include the inverse piezoelectric tensor (the electrostriction tensor) \cite{Ha-07}, third order symmetric traceless-tensors in liquid crystal study \cite{BPC-02, CQV-17, GV-16, Fe-95, LR-02, Vi-15}, third order susceptibility tensors in nonlinear optics study \cite{Je-70, KGTRU-04}, and the Hall tensor \cite{Ha-07}.   On the other hand, the Levi-Civita tensor is famous in tensor calculus \cite{Ha-07, LCE-10}.

Under an orthonormal basis, a third order tensor is represented by a hypermatrix $(a_{ijk})$.   While matrix theory and second order tensor theory are well-developed, the theory on third order tensors and hypermatrices is still less developed.    Some properties of hypermatrices are not preserved under orthonormal transformations.   Thus, these properties are not tensor properties.  For example, the nonnegativity property that $a_{ijk} \ge 0$ is not preserved under orthonormal transformations.    Thus, it is not a tensor property.   Recently, a  comprehensive Perron-Frobenius theory has been developed for nonnegative hypermatrices and has applications in spectral hypergraph theory and higher order Markov chains \cite{QL-17}.  However, since it is not a tensor property, it may be useless or less useful in mechanics and some areas of physics.

Hence, while it is a common practice to use a hypermatrix symbol $a_{ijk}$ to denote a third order tensor in physics and mechanics, a systematic study on third order tensors and hypermatrices, to distinguish them clearly, is necessary.

Essentially, a second order tensor is a linear operator transforming a first order tensor (vector) into another first order tensor (vector), while a third order tensor is a linear operator which transforms a second order tensor into a first order tensor, or a first order tensor into a second order tensor.   For example, the piezoelectric effect can be displayed as
$$\x = \A\V,$$
where $\x$ is the vector of the electric polarization, $\A$ is the piezoelectric tensor, and $\V$ is the mechanical stress tensor \cite{Ha-07}.   Thus, the third order tensor $\A$ transforms a second order tensor $\V$ to a first order tensor $\x$.    On the other hand, the inverse piezoelectric effect can be displayed as
$$\U = \B\z,$$
where $\U$ is the deformation, $\B$ is the inverse piezoelectric tensor, and $\z$ is the electric field strength \cite{Ha-07}.   Thus, the third order tensor $\B$ transforms a first order tensor $\z$ to a second order tensor $\U$.

In this paper, from this point of view, we study third order tensors and (third order) hypermatrices systematically, and introduce some new concepts for third order tensors, such as transposes, kernel tensors, L-inverses, L-eigenvalues, singular values, and third order orthogonal tensors.

We assume that the dimension is $3$.  In mechanics, tensors are in the three dimensional space.  It is easy to extend most results to higher dimensions.

In the next section, we review knowledge about first and second order tensors and their representative matrices.

We start to study third order tensors and hypermatrices in Section 3.  We study various operations and symmetric and anti-symmetric properties of third order tensors, introduce transposes of third order tensors, and third order orthogonal tensors in that section.   The transpose of a third order tensor is uniquely defined.  By applying the transpose operation three times to a third order tensor, we recover that third order tensor itself.   A third order tensor is called a third order orthogonal tensor if its (second order tensor) product with its transpose is the second order identity tensor.  We also define non-singularity of a third order tensor in that section.

In Section 4 we introduce the kernel tensor and the L-inverse of a third order tensor.
The kernel tensor of a third order tensor is a second order positive semi-definite symmetric tensor, which is the product of that third order tensor and its transpose.     A third order tensor has an L-inverse if and only if it is nonsingular.    The L-inverse of a third order nonsingular tensor is always unique.    If a third order tensor $\A$ maps a first order tensor $\z$ into a second order tensor $\U$, then its L-inverse $\A^{-1}$ can recover $\z$ from $\U$ in the operation $\z = \A^{-1}\U$.   A typical example is for the inverse piezoelectric effect.  Here, the letter ``L'' is named after Levi-Civita.

Then, in Section 5, we define L-eigenvalues for a third order tensor.  The letter ``L'' here is also named after Levi-Civita.   A third order tensor has three L-eigenvalues, with associated L-eigenvectors and L-eigentensors.   L-eigentensors are second order tensors.   Three L-eigenvalues are all nonnegative.    The third order tensors are nonsingular if and only if all of its L-eigenvalues are positive.  L-eigenvectors and L-eigentensors of different L-eigenvalues are orthogonal to each other, respectively.    The three L-eigenvalues of a third order tensor are square roots of three eigenvalues of its kernel tensor.   The L-eigenvectors of third order tensor are corresponding eigenvectors of the kernel tensor of the third order tensor.   From these we may compute the corresponding L-eigenvectors of the third order tensor.    Using these we may compute the L-inverse of the third order tensor if all of its L-eigenvalues are positive.

For a second order symmetric tensor, it is well-known that it has an eigenvector decomposition and three principal invariants \cite{HXL-04}.  In Section 6, for a third order partially symmetric tensor, we also give its eigenvector decomposition.  The piezoelectric tensor, the inverse piezoelectric tensor, and third order symmetric tensors are third order partially symmetric tensors.  We make a conjecture a third order symmetric tensor has seven independent invariants.  We raise the questions on how may independent invariants a third order tensor, a third order partially symmetric tensor and a third order cyclically symmetric tensor may have respectively.

In Section 7, we introduce singular values, C-eigenvalues and Z-eigenvalues for a third order tensor.   As L-eigenvalues, they are all invariants of that third order tensor.

We study the Levi-Civita tensor and related third order tensors in Section 8.  We show that the Levi-Civita tensor is nonsingular, its L-inverse is a half of itself, and its three L-eigenvalues are all the square root of two.

Some final remarks are made in Section 9.

\section{First and Second Order Tensors}

In this section, we review some preliminary knowledge on first and second order tensors.

We use $x_i, y_i, z_i, \cdots$, to denote vectors in $\Re^3$, and $u_{ij}, v_{ij}, \cdots$, to denote $3 \times 3$ matrices in $\Re^3$, where the indices $i$ and $j$ go from $1$ to $3$.   In a product of vectors and matrices, if one index is repeated twice, it means summation on this index from $1$ to $3$.   This is the common usage in the literature of physics.   Actually, if we single out one vector or one matrix, their indices can be substituted with any other ones.  But if we consider a product or a relation, or an equality of vectors and matrices, the indices there are related, and cannot be changed arbitrarily.

We use $\delta_{ij}$ for the Kronecker symbol.   In the above hypermatrix notation, $\delta_{ij}$ denotes the identity matrix.     A matrix $p_{ij}$ is called an orthogonal matrix if it satisfies
$$p_{ik}p_{jk} = \delta_{ij}.$$

We use small bold letters $\x, \y, \z, \cdots$, to denote first order tensors in a three dimensional physical space $\HH$.  First order tensors may also be called vectors.   To distinguish them from vectors $x_i, y_i, z_i, \cdots$ in $\Re^3$, in the most places of this paper, we only call them first order tensors.  First order tensors have additions among themselves and multiplications with scalars.  Therefore, there is a first order zero tensor, which we denote as $\0$.    Assume that for any $\x, \y \in \HH$, there is a scalar, denoted as $\x \bullet \y$, which satisfies the inner product rules, and $\HH$ is complete regarded to this inner product.    Then, mathematically,
$\HH$ is a Hilbert space.   Thus, we call $\x \bullet \y$ the inner product of $\x$ and $\y$.   If $\x \bullet \y = 0$, then we say that $\x$ and $\y$ are orthogonal.  If $\x \bullet \x =1$, then we say that $\x$ is a first order unit tensor.    Suppose that we have an orthonormal basis $\{ \e_1, \e_2, \e_3 \}$.  Under this basis, two first order tensors $\x$ and $\y$ in $\HH$ are represented by two vectors $x_i$ and $y_i$ in $\Re^3$, respectively.   Then we have
$$\x \bullet \y = x_iy_i.$$
The summation value $x_iy_i$ is actually independent from the choice of the basis.
In particular, the first order zero tensor $\0$ is always represented by the zero vector and vice versa.

We use capital bold letters $\U, \V, \cdots$, to denote second order tensors in $\HH$.  Denote the set of all second order tensors in $\HH$ by $B(\HH)$.  Under an orthonormal basis $\{ \e_1, \e_2, \e_3 \}$, a second order tensor $\U$ in $\HH$ is represented by a matrix $u_{ij}$ in $\Re^3$.  Assume that a first order tensor $\x$ is represented by $x_i$ under this basis.    Suppose that we have another orthonormal basis $\{ \g_1, \g_2, \g_3 \}$.  Under that basis, assume that $\x$ is represented by a vector $z_q$ in $\Re^3$, and $\U$ is represented by a matrix $v_{qr}$ in $\Re^3$.  Then there is an orthogonal matrix $p_{iq}$, which is only dependent on these two bases, such that $x_i = p_{iq}z_q$ and $u_{ij} = p_{iq}p_{jr}v_{qr}$.

Second order tensors are physical quantities.   They also have additions among themselves and multiplications with scalars.   Thus, there is a second order zero tensor, which we denote as $\OO$. Furthermore, second order tensors may be applied to first order tensors, two first order tensors may have a second order tensor product, and two second order tensors may have products.   Let $\U, \V, \W \in B(\HH)$, $\x, \y \in \HH$.   Assume that they are represented by $u_{ij}, v_{ij}, w_{ij}, x_i$ and $y_j$ under an orthonormal basis.   These operations can be described as follows.

1. $\x \U = \y$, represented by $x_iu_{ij} = y_j$.

2. $\x = \U \y$, represented by $x_i = u_{ij}y_j$.

3. $\U = \x \otimes \y$, represented by $u_{ij} = x_iy_j$.   Such a second order tensor $\U = \x \otimes \y$ is called a second order rank-one tensor.

4. $\x\U \y = x_iu_{ij}y_j$.  This value is independent from the choice of the basis.

5. $\U \bullet \V  = u_{ij}v_{ij}$.  This value is also independent from the choice of the basis, and called the inner product of $\U$ and $\V$.

6. $\W =\U \V$, represented by $w_{ij} = u_{ik}v_{kj}$.   The second order tensor $\W$ is called the product of $\U$ and $\V$.

A second order tensor $\U \in B(\HH)$ is called singular if there is $\y \in \HH, \y \not = \0$ such that $\U\y = \0$.   Otherwise, $\U$ is called nonsingular.

For any $\U \in B(\HH)$, there is a unique second order tensor $\U^\T \in B(\HH)$ such that for for any $\x, \y \in \HH$, $\x \U \y = \y \U^\T \x$.    The second order tensor $\U^\T$ is called the transpose of $\U$.   Clearly, under an orthonormal basis, if $\U$ is represented by $u_{ij}$, then $\U^\T$ must be represented by $u_{ji}$.   For any $\U \in B(\HH)$, we have
 $(\U^\T)^\T = \U$.    A second order tensor $\U$ is nonsingular if and only if $\U^\T$ is nonsingular.   If $\U = \U^\T$, then $\U$ is called symmetric. If $\U = - \U^\T$, then $\U$ is called anti-symmetric.   They are always represented by symmetric matrices and anti-symmetric matrices respectively, under an orthonormal basis.

 If $\U = \U \V$ for any $\U \in B(\HH)$, then $\V$ is a special tensor, called the second order identity tensor, and denoted by $\I$.   For any $\U \in B(\HH)$, we also have $\U = \I \U$.   The identity tensor $\I$ is always represented by $\delta_{ij}$.

 If $\U\V = \I$, then $\V$ is called the inverse of $\U$, and denoted as $\U^{-1}$.   A second order tensor $\U$ has an inverse if and only if it is nonsingular, the inverse is unique if it exists, and we always have $(\U^{-1})^{-1} = \U$.

 If $\U\U^\T = \I$, then $\U$ is called an orthogonal tensor.   Thus, an orthogonal tensor $\U$ is always nonsingular, and $\U^{-1} = \U^\T$.   A second order tensor is orthogonal if and only if its representative matrix under an orthonormal basis is orthogonal.

Since the value of $\U \bullet \V$ is independent from the choice of the basis, for any $\U \in B(\HH)$, $\U \bullet \I$ is independent from the choice of the basis.   This value is called the trace of $\U$, represented by $u_{ii}$ under an orthonormal basis.   A second order tensor is called a second order traceless tensor if its trace is zero.

If $\U \bullet \V = 0$, then we say that $\U$ and $\V$ are orthogonal.    If $\U \bullet \U = 1$, then we say that $\U$ is a second order unit tensor.

If $\U \y = \lambda \y$ for $\y \not = \0$, then $\lambda$ is called an eigenvalue of $\U$ and $\y$ is called an eigenvector of $\U$, associated with the eigenvalue $\lambda$.  The eigenvalue of $\U$ is independent from the choice of the basis.   A second order tensor is nonsingular if and only if it has no zero eigenvalue.     A second order symmetric tensor $\U$ has three real eigenvalues $\lambda_1 \ge \lambda_2 \ge \lambda_3$.   Furthermore, there are three eigenvectors $\y_1, \y_2$ and $\y_3$ associated with  $\lambda_1, \lambda _2$ and $\lambda_3$, which are orthogonal to each other, i.e., $\U \y_i = \lambda_i \y_i$,  and $\y_i \bullet \y_j = \delta_{ij}$ for $i, j = 1, 2, 3$.  Furthermore, such a second order symmetric tensor $\U$ has an eigenvector decomposition
\begin{equation} \label{decomp}
\U = \sum_{j=1}^3 \lambda_j  \y_j \otimes \y_j.
\end{equation}

Let $\U \in B(\HH)$.    Then $\{ \y \in \HH : \U\y = \0 \}$ is a linear subspace of $\HH$.  We call this subspace the null space of $\U$, and define the rank of $\U$ as $3$ minus the dimension of the null space of $\U$.

The eigenvalues of a second order tensor are invariants of that tensor.   For a second order symmetric tensor $\U$, its trace is equal to the sum of its eigenvalues, the determinant of any of its representative matrices is equal to the product of its eigenvalues.  Hence, the determinant of any of its representative matrices is also an invariant of that tensor, and can be called the determinant of that tensor.    This is also true for second order nonsymmetric tensors.  The discussion of this involves complex eigenvalues.   For a second order tensor $\U$, we denote its determinant by det$(\U)$ and its trace by tr$(\U)$.   Then $\phi(\lambda) \equiv$ det$(\lambda \I - \U)$ is a one dimensional cubic polynomial.   We call $\phi(\lambda)$ the characteristic polynomial of $\U$. The three roots of $\phi(\lambda)$ are exactly three eigenvalues of $\U$.

Let $\U$ be a second order symmetric tensor.  Then it has three independent invariants tr$(\U)$, tr$(\U^2)$ and tr$(\U^3)$ \cite{Zh-94}.

Suppose $\U$ is a second order symmetric tensor.   We say that $\U$ is positive semi-definite if for any $\x \in \HH$, $\x \U \x \ge 0$.   We say that $\U$ is positive definite if for any $\x \in \HH$ and $\x \not = \0$, we have $\x \U \x > 0$.    A second order symmetric tensor is positive semi-definite if and only if all of its eigenvalues are nonnegative.    A second order symmetric tensor is positive definite if and only if all of its eigenvalues are positive.

These are the main content of the theory of first and second order tensors.

\section{Third Order Tensors and Hypermatrices}

We now study third order tensors and (third order) hypermatrices.

We use $a_{ijk}, b_{ijk}, \cdots$, to denote $3 \times 3 \times 3$ hypermatrices in $\Re^3$, where the indices $i, j$ and $k$ go from $1$ to $3$.     Again, in a product of vectors, matrices and hypermatrices, if one index is repeated twice, it means summation on this index from $1$ to $3$.

We call a hypermatrix $a_{ijk}$ an orthogonal hypermatrix if
$$a_{ijk}a_{ljk} = \delta_{il}.$$

We use calligraphic letters $\A, \B, \cdots$, to denote third order tensors in $\HH$, and denote the set of third order tensors in $\HH$ by $T(\HH)$.   When we discuss the products of two third order tensors, we also need to refer to a fourth order tensor, denoted by the calligraphic letter $\Ten$.

 Suppose that we have an orthonormal basis $\{ \e_1, \e_2, \e_3 \}$.  Under this basis, a first order tensor $\x$ in $\HH$ is represented by a vector $x_i$ in $\Re^3$, a second order tensor $\U$ in $\HH$ is represented by a matrix $u_{ij}$ in $\Re^3$, a third order tensor $\A$ is represented by a hypermatrix $a_{ijk}$ in $\Re^3$.   Suppose that we have another orthonormal basis $\{ \g_1, \g_2, \g_3 \}$.  Under that basis, assume that $\x$ is represented by a vector $y_q$ in $\Re^3$, $\U$ is represented by a matrix $v_{qr}$ in $\Re^3$, $\A$ is represented by a hypermatrix $b_{qrs}$ in $\Re^3$.  Then there is an orthogonal matrix $p_{iq}$, which is only dependent on these two bases, such that $x_i = p_{iq}y_q$, $u_{ij} = p_{iq}p_{jr}v_{qr}$, and $a_{ijk} = p_{iq}p_{jr}p_{ks}b_{qrs}$.

 Third order tensors are also physical quantities.   They also have additions among themselves and multiplications with scalars.   Thus, there is a third order zero tensor, which we denote as $\OOO$. Furthermore, third order tensors may be applied to first order and second order tensors, one first order tensor and one second order tensor, or three first order tensors may have a third order tensor product, and two third order tensors may also have products.   Let $\A, \B \in T(\HH), \U, \V \in B(\HH)$, $\x, \y, \z \in \HH$, and a fourth order tensor $\Ten$ be represented by $a_{ijk}, b_{ijk}, u_{ij}, v_{jk}, x_i, y_j, z_k$ and $t_{ijkl}$ under an orthonormal basis.   Some of these operations can be described as follows.

1. $\x\A = \V$, represented by $x_ia_{ijk}= v_{jk}$.

2. $\U = \A\z$, represented by $u_{ij}=a_{ijk}z_k$.

3. $\U\A = \z$, represented by $u_{ij}a_{ijk} = z_k$.

4. $\x = \A\V$, represented by $x_i = a_{ijk}v_{jk}$.

5. $\x\y\A = \z$, represented by $x_iy_ja_{ijk} = z_k$.

6. $\x\A\z = \y$, represented by $x_ia_{ijk}z_k = y_j$.

7. $\x = \A\y\z$, represented by $x_i = a_{ijk}y_jz_k$.

8. $\A = \U \otimes \z$, represented by $a_{ijk} = u_{ij}z_k$.

9. $\A = \x \otimes \V$, represented by $a_{ijk} = x_iv_{jk}$.

10. $\A = \x \otimes \y \otimes \z$, represented by $a_{ijk} = x_iy_jz_k$.  Such a third order tensor $\A = \x \otimes \y \otimes \z$ is called a third order rank-one tensor.

11. $\x\A\y\z = x_ia_{ijk}y_jz_k$.   This value is independent from the choice of the basis.

12. $\A \bullet \B = a_{ijk}b_{ijk}$.   This value is also independent from the choice of the basis, and called the inner product of $\A$ and $\B$.

13. $\U = \A \B$, represented by $u_{il} = a_{ijk}b_{jkl}$.   The tensor $\U$ is called the second order tensor product of $\A$ and $\B$.

14. $\Ten = \A \oplus \B$, represented by $t_{ijkl} = a_{ijm}b_{mkl}$.    The tensor $\Ten$ is called the fourth order tensor product of $\A$ and $\B$.

Suppose that $\A \in T(\HH)$ is represented by hypermatrix $a_{ijk}$ under an orthonormal basis $\{ \e_1, \e_2, \e_3 \}$ of $\HH$.   Then
$$a_{ijk} = \e_i\A\e_j\e_k.$$
Note that $\{-\e_1, -\e_2, -\e_3 \}$ is also an orthonormal basis of $\HH$.   Under that orthonormal basis, $\A$ is represented by hypermatrix
$$-\e_i\A(-\e_j)(-\e_k) = - a_{ijk}.$$
This property of third order tensors is significant.   In a crystal with a reflection symmetry, a third order tensor must be zero.    Thus, there is no piezoelectric effect in such a crystal
\cite{Ha-07, KPG-08, Lo-89}.

A third order tensor $\A \in T(\HH)$ is called singular if there is $\x \in \HH$, $\x \not = \0$ such that $\x\A = \OO$.  Otherwise, $\A$ is called nonsingular.

A third order tensor $\A \in T(\HH)$ is called right-side symmetric if for any $\y, \z \in \HH$, $\A\y\z = \A\z\y$.   Under an orthonormal basis, if $\A$ is represented by $a_{ijk}$, then we have $a_{ijk} = a_{ikj}$ for all $i, j$ and $k$.    A typical example of a third order right-side symmetric tensor is the piezoelectric tensor in solid crystals \cite{CJQ-17, CC-80, Ha-07, KPG-08, Lo-89, ZTP-13}.   On the other hand, in liquid crystals, the piezoelectric tensor is not right-side symmetric \cite{Ja-10, JPSRd-09, JTSSS-02}.

A third order tensor $\A \in T(\HH)$ is called left-side symmetric if for any $\x, \y \in \HH$, $\x\y\A = \y\x\A$.   Under an orthonormal basis, if $\A$ is represented by $a_{ijk}$, then we have $a_{ijk} = a_{jik}$ for all $i, j$ and $k$.    A typical example of a third order left-side symmetric tensor is the inverse piezoelectric tensor \cite{Ha-07}.

A third order tensor $\A \in T(\HH)$ is called centrally symmetric if for any $\x, \z \in \HH$, $\x\A\z = \z\A\x$.   A third order tensor is called partially symmetric if it is either right-side symmetric, or left-side symmetric, or centrally symmetric.   Clearly, if a third order tensor is both right-side and left-side symmetric, then it is also centrally symmetric.  We call such a third order tensor a third order symmetric tensor.    Obviously, a third order tensor is symmetric if it is both right-side and centrally symmetric, or if it is both left-side and centrally symmetric.    The piezoelectric tensor, the inverse piezoelectric tensor, and third order symmetric tensors are third order partially symmetric tensors.

A third order tensor $\A \in T(\HH)$ is called left-side anti-symmetric if for any $\x, \y \in \HH$, $\x\y\A = -\y\x\A$.   Under an orthonormal basis, if $\A$ is represented by $a_{ijk}$, then we have $a_{ijk} = -a_{jik}$ for all $i, j$ and $k$.    A typical example of a third order left-side anti-symmetric tensor is the Hall tensor \cite{Ha-07}.

A third order tensor $\A \in T(\HH)$ is called right-side anti-symmetric if for any $\y, \z \in \HH$, $\A\y\z = -\A\z\y$.  A third order tensor $\A \in T(\HH)$ is called centrally anti-symmetric if for any $\x, \z \in \HH$, $\x\A\z = -\z\A\x$.   A third order tensor $\A \in T(\HH)$ is called totally anti-symmetric if it is right-side, left-side and centrally anti-symmetric.

A famous example of third order totally anti-symmetric tensors is the Levi-Civita tensor, or called the permutation tensor, which we denote as $\E$.  Under a certain orthonormal basis, $\E$ is represented by the Levi-Civita hypermatrix $\epsilon_{ijk}$ with $\epsilon_{123}= \epsilon_{312} = \epsilon_{231} = 1$, $\epsilon_{213} = \epsilon_{321} = \epsilon_{132} = -1$ and $\epsilon_{ijk} = 0$ otherwise.    Other third order totally anti-symmetric tensors are multiples of the Levi-Civita tensor.   Then, we see that the representative hypermatrix of the the Levi-Civita tensor $\E$ is either $\epsilon_{ijk}$ or $-\epsilon_{ijk}$.   A question is: is there any other third order tensor, except multiples of the Levi-Civita tensor, which is represented by either a fixed hypermatrix $a_{ijk}$, or $-a_{ijk}$?

Suppose that $\A, \B \in T(\HH)$ satisfy $\x\A\y\z = \y\B\z\x$ for any $\x, \y, \z \in \HH$.   Then we call $\B$ the transpose of $\A$ and denote that $\B = \A^\T$.    In particular, thermodynamical relations demand that the components of the piezoelectric and electrostriction tensors take on the same numerical value.  Hence, the transpose of the piezoelectric tensor is the inverse piezoelectric tensor (the electrostriction tensor) \cite{Ha-07}.

\begin{Proposition} \label{p1}
Suppose $\A \in T(\HH)$.  Then its transpose always exists and is unique, and we
always have $[(\A^\T)^\T]^\T = \A$.
\end{Proposition}
\begin{proof}   Let $\A$ be represented by $a_{ijk}$ under an orthonormal basis.    Let $\B$ be represented by $a_{jki}$ under this basis.  Then we may easily see that $\B = \A^\T$.   This proves the existence.

Suppose that $\C$ also satisfies the definition of the transpose of $\A$.   Then for any $\x, \y, \z \in \HH$, we have $\y(\B-C)\z\x = 0$.   This implies that $\B - \C = \OOO$, i.e., $\B = \C$.  This proves the uniqueness.

The third conclusion is from the definition directly.
\end{proof}

If $\A^\T = \A$, then we say that $\A$ is cyclically symmetric.  The Levi-Civita tensor is a cyclically symmetric tensor.     It is also easy to see that a third order tensor is symmetric if and only if it is both partially and cyclically symmetric.

If $\A\A^\T = \I$, then we call $\A$ a third order orthogonal tensor.   Since the representative matrix of $\I$ is always $\delta_{ij}$, we have the following proposition.

\begin{Proposition} \label{p1-1}
Suppose $\A \in T(\HH)$.  Then $\A$ is orthogonal if and only its representative hypermatrix under an orthonormal basis is orthogonal.
\end{Proposition}

Let $\A \in T(\HH)$ be symmetric.  We say that $\A$ is traceless if $\A\I = \0$.   Third order symmetric traceless tensors play an important role in the study of liquid crystals.   It was theoretically proposed \cite{Fe-95} and experimentally confirmed \cite{BPC-02,LR-02} that the phases in liquid crystals
composed of bent-core molecules could be described by means of a third order symmetric traceless tensor.

In the next three sections, we study more properties of third order tensors.    An important fact we will use in our discussion is that the tensor contraction relations of representative hypermatrices, matrices and vectors of tensors are independent from the choice of basis.

\section{The Kernel Tensor and The L-Inverse of A Third Order Tensor}

Let $\A \in T(\HH)$.  Let $\U = \A \A^\T$.  Then we call $\U$ the kernel tensor of $\A$.

\begin{Theorem} \label{t1}
Let $\A \in T(\HH)$ and $\U \in B(\HH)$ be its kernel tensor.   Then $\U$ is symmetric and positive semi-definite.    Furthermore, $\U$ is positive definite if and only if $\A$ is nonsingular.  In particular, a third order orthogonal tensor is nonsingular.
\end{Theorem}
\begin{proof}  Suppose that under an orthonormal basis, $\A$ is represented by $a_{ijk}$.   Then its kernel tensor $\U$ is represented by $u_{il} = a_{ijk}a_{jkl}$, a symmetric and positive semi-definite matrix.   Thus, $\U$ is symmetric and positive semi-definite.

If $\A$ is singular, then there is $\x \in \HH$, $\x \not = \0$ such that $\x \A = \0$.  Then $\x \U \x = (\x \A)(\A^\T \x) = 0$, i.e., $\U$ cannot be positive definite.

For indices $(j, k) = (1, 1), (1, 2), (1, 3), (2, 1), (2, 2), (2, 3), (3, 1), (3, 2), (3, 3)$, we use a single index $r = 1, \cdots, 9$ to represent them.   Then we may regard hypermatrix $a_{ijk}$ as a $3 \times 9$ matrix $A = (a_{ir})$.  Then the matrix $U = (u_{il}) = AA^\T$.  By the singular value decomposition theory of rectangular matrices \cite{WWQ-04}, we have $A = S\Sigma T$, where $S$ is a $3 \times 3$ orthogonal matrix, $T$ is a $9 \times 9$ orthogonal matrix, $\Sigma$ is a $3 \times 9$ diagonal matrix with diagonal entries $\sigma_1 \ge \sigma_2 \ge \sigma_3 \ge 0$, which are singular values of $\A$, and $\sigma_1^2, \sigma_2^2, \sigma_3^2$ are eigenvalues of $U$.   Suppose that $\A$ is nonsingular.  Then $\Sigma$ must have full row rank $3$.  This implies that $\sigma_3 >0$ and thus, $\sigma_3^2 > 0$.   Then, $U$ must be positive definite.  This implies that $\U$ is positive definite.

For a third order orthogonal tensor $\A$, we have $\A\A^\T = \I$ by definition.   Since $\I$ is nonsingular, $\A$ is also nonsingular.
\end{proof}

The kernel tensor $\U = \A\A^\T$ is uniquely defined by $\A$.    Thus, the invariants of $\U$, such as the trace and the determinant of $\U$ are also invariants of $\A$.   We may call the rank of $\U$ as the rank of $\A$.

Let $\A$ and $\B$ be two third order tensors in $\HH$. We say that $\B$ is the L-inverse of $\A$ and denote that $\B = \A^{-1}$, if $\A$ and $\B$ satisfy
\begin{equation} \label{e1}
\A\B = \I
\end{equation}
and
\begin{equation} \label{e2}
\B \oplus \A = \A^\T \oplus (\B^\T)^\T.
\end{equation}

\begin{Theorem} \label{t2}
Suppose that a third order tensor $\A$ in $\HH$ has an L-inverse $\A^{-1}$.   Then such an L-inverse is unique, and $(\A^{-1})^{-1} = \A$.

Furthermore, $\A$ has an L-inverse if and only if $\A$ is nonsingular.

If $\A\A^\T = \alpha\I$ for some $\alpha > 0$, then $\A$ is nonsingular, and $\A^{-1} = {1\over \alpha}\A^\T$.   In particular, if $\A$ is a third order orthogonal tensor, then $\A^{-1} = \A^\T$.
\end{Theorem}
\begin{proof}  Suppose that $\A \in T(\HH)$ has an inverse $\B = \A^\T$.   Assume that under an orthonormal basis, $\A$ and $\B$ are represented by $a_{ijk}$ and $b_{jkl}$ respectively.  For indices $(j, k) = (1, 1), (1, 2), (1, 3), (2, 1)$, $(2, 2), (2, 3), (3, 1), (3, 2), (3, 3)$, we use a single index $r = 1, \cdots, 9$ to represent them.   Then we may regard hypermatrix $a_{ijk}$ as a $3 \times 9$ matrix $A = (a_{ir})$, and  hypermatrix $b_{jkl}$ as a $9 \times 3$ matrix $B = (b_{rl})$, respectively.  By the generalized inverse theory of matrices \cite{WWQ-04}, matrix $A$ has a unique Moore-Penrose inverse, and matrix $B$ is such a unique Moore-Penrose inverse of $A$ and vice versa, if and only if

1. $BAB = B$;

2. $ABA = A$;

3. $AB = (AB)^\T$;

4. $BA = (BA)^\T$.

When (\ref{e1}) and (\ref{e2}) holds, the above four requirements are satisfied.  Thus, $B$ is the unique Moore-Penrose inverse of $A$ and vice versa.    This shows that such an L-inverse is unique if it exists, and we have $(\A^{-1})^{-1} = \A$.  Furthermore, in this case, the product $AB$ is the $3 \times 3$ identity matrix.   Then $A$ must has full row rank.  This implies that $\A$ is nonsingular.

Suppose that $\A$ is nonsingular and is represented by $a_{ijk}$ under an orthonormal basis.   Let $A$ be the $3 \times 9$ matrix constructed as above.    Then $A$ has a singular value decomposition $A=S\Sigma T$ as in the proof of Theorem \ref{t1}.  Then $\Sigma$ is a $3 \times 9$ diagonal matrix, with diagonal entries $\sigma_1 \ge \sigma_2 \ge \sigma_3 > 0$, as $\A$ is nonsingular.   Let $B = (b_{rl})$ be the Moore-Penrose inverse of $A$.  Let $b_{jkl}$ be the hypermatrix recovered from $B$.   Let $\B$ be the third order tensor represented by $b_{ijk}$.   Then we may verify that $\B = \A^{-1}$.

Suppose that $\A\A^\T = \alpha\I$ for some $\alpha > 0$.   Let $\B = {1 \over \alpha}\A^\T$.    Then $\A\B = \I$, and
$$\A^\T \oplus (\B^\T)^\T = \A^\T \oplus {1 \over \alpha}\A = \B \oplus \A,$$
i.e, (\ref{e1}) and (\ref{e2}) are satisfied.  Thus, $\A^{-1} =\B = {1 \over \alpha}\A^\T$.
\end{proof}

The concept of L-inverse is somehow connected with the Moore-Penrose inverse of a rectangular matrix.  But they are different.   The Moore-Penrose inverse of a rectangular always exists.   The L-inverse of a third order tensor exists if and only if the third order tensor is nonsingular.  Second, the L-inverse is a tensor concept.  It is associated with the transformation from a first order tensor to a second order tensor.   The Moore-Penrose inverse of a rectangular matrix has no such a property.   The following application of L-inverses further confirms this.

\begin{Proposition} \label{p1}
Suppose that we have tensor relation
$$\x\A = \V$$
and $\A^{-1}$ exists.  Then we have
$$\x = \V\A^{-1}.$$
\end{Proposition}
\begin{proof} We have
$$\V\A^{-1} = (\x\A)\A^{-1} = \x(\A\A^{-1}) = \x\I = \x.$$
\end{proof}

A typical example of $\x\A = \V$ is the inverse piezoelectric effect, where $\A$ is the piezoelectric tensor, $\x$ is the electric field strength and $\V$ is the deformation \cite{Ha-07, KPG-08}.  Thus, if we know $\A^{-1}$, we may calculate $\x$ from $\V$.   Note that $\x\A= \V$ is equivalent to $\V = \B\x$, where $\B$ is the inverse piezoelectric tensor, as $\B = \A^\T$.

A question is: what is the relation between the inverse piezoelectric tensor $\B$ and the L-inverse of the piezoelectric tensor $\A$?   If $\B \equiv \A^\T = \A^{-1}$, then by Theorem \ref{t2}, the piezoelectric tensor should be an orthogonal tensor.

\section{L-Eigenvalues of a Third Order Tensor}

We now define L-eigenvalues for a third order tensor.

Let $\A \in T(\HH), \V \in B(\HH), \x \in \HH$ and $\sigma \in \Re$.    We say that $\sigma$ is an L-eigenvalue of $\A$, $\V$ and $\x$ are associated L-eigentensor and L-eigenvector respectively if $\sigma \ge 0$ and the following equations hold.
\begin{equation} \label{e3}
\A\V = \sigma \x, \ \ \ \A^\T\x = \sigma \V, \ \ \ \V \bullet \V = 1, \ \ \  \x \bullet \x = 1.
\end{equation}
We count the multiplicity of an L-eigenvalue of $\A$ by the number of linearly independent L-eigenvectors.

\begin{Theorem} \label{t3}
Suppose that $\A \in T(\HH)$.   Then $\A$ has three L-eigenvalues $\sigma_1 \ge \sigma_2 \ge \sigma_3 \ge 0$, with associated L-eigentensors $\V_1, \V_2, \V_3$, and L-eigenvectors $\x_1, \x_2, \x_3$.   They have the following properties:

1. $\V_i \bullet \V_j = \delta_{ij}$ and $\x_i \bullet \x_j = \delta_{ij}$.

2. If $\U = \A\A^\T$ is the kernel tensor of $\A$, with eigenvalues $\lambda_1 \ge \lambda_2 \ge \lambda_3$, then $\sigma_j^2 = \lambda_j$ for $j = 1, 2, 3$, and $\x_j$, $j = 1, 2, 3$ are associated eigenvectors of $\U$.   Hence L-eigenvalues of $\A$ are invariants of $\A$.

3. We have
\begin{equation} \label{e4}
\A = \sum_{j=1}^3 \sigma_j \x_j \otimes \V_j.
\end{equation}

4. $\A^{-1}$ exists if and only if $\sigma_1 \ge \sigma_2 \ge \sigma_3 > 0$.    In this case, we have
\begin{equation} \label{e5}
\A^{-1} = \sum_{j=1}^3 {1 \over \sigma_j} \V_j \otimes \x_j.
\end{equation}

5. We have
\begin{equation} \label{e6}
\sigma_1 = \max \{ \sqrt{(\A \V) \bullet (\A \V)} : \V \bullet \V = 1 \}.
\end{equation}

6. If $\A$ is right-side symmetric, then all the L-eigentensors associated with its positive L-eigenvalues are symmetric.

7. If $\A\A^\T = \alpha \I$ for some $\alpha \ge 0$, then the L-eigenvalues of $\A$ are $\sigma_1 = \sigma_2 = \sigma_3 = \sqrt{\alpha}$.  In particular, if $\A$ is a third order orthogonal tensor, then its L-eigenvalues are $\sigma_1 = \sigma_2 = \sigma_3 = 1$.
\end{Theorem}
\begin{proof}  For 1-5, we may use the unfolding technique in the proofs of Theorems \ref{t1} and \ref{t2}, the singular value decomposition theory \cite{WWQ-04} to prove them.  We do not go to the details.  Conclusion 6 follows from $\A^\T\x = \sigma\V$ in (\ref{e3}) and $\sigma \not = 0$ directly.   Since the eigenvalues of $\alpha\I$ are $\lambda_1 = \lambda_2 = \lambda_3 = \alpha$, we have conclusion 7.
\end{proof}

If $\A\A^\T = \alpha \I$ for some $\alpha \ge 0$, then the three L-eigenvalues of $\A$ are the same.   On the other hand, if the three L-eigenvalues of a third order tensor $\A$ are the same, do we always have $\A\A^\T = \alpha \I$ for some $\alpha \ge 0$?

We call (\ref{e4}) and (\ref{e5}) the L-eigenvalue decomposition of $\A$ and its L-inverse.  Note that we do not say that $1 \over \sigma_j$ for $j = 1, 2, 3$ are L-eigenvalues of $\A^{-1}$, as $\A^{-1}$ may have different L-eigenvalues.

We do not call $\sigma_j, j = 1, 2, 3$ singular values, though they are associated with the singular value decomposition theory of rectangular matrices.   One reason is that they are third order tensors, not rectangular matrices.  They are only associated with rectangular matrices unfolded from the representative hypermatrices unfolded with respect to the last two indices.   If we unfold the related hypermatrices with respect to the other two indices, the results may be different.   The second reason is that we will define singular values for third order tensors later, which are different.

Equation (\ref{e6}) has clear physical meanings.  Suppose that $\A$ is the piezoelectric tensor and $V$ is the stress tensor.   Then (\ref{e6}) says that the largest L-eigenvalue of the piezoelectric tensor gives the largest magnitude of the electric polarization vector under unit stress tensors.

\bigskip

For a third order tensor $\A \in T(\HH)$, $\{ \V \in B(\HH) : \A\V = \0 \}$ is a linear subspace of $B(\HH)$.   We call it the null space of $\A$.  By the unfolding technique in the proofs of Theorems \ref{t1} and \ref{t2}, we have the following proposition.

\begin{Proposition} \label{p3}
Let $\A \in T(\HH)$.   The dimension of its null space is at least $6$.    The sum of its rank and the dimension of its null space is $9$.
\end{Proposition}

This proposition implies that the null space of a third order tensor is quite ``large'', as its dimension is at least $6$.   Suppose that $\A$ is the piezoelectric tensor.   It is important to identify its null space.   If the stress tensor falls in the null space of the piezoelectric tensor, then there is no piezoelectric effect.

\section{Eigenvector Decomposition and Invariants}

For third order partially symmetric tensors, we have the following theorem which is stronger than the sixth conclusion of Theorem \ref{t3}.

\begin{Theorem} \label{t3.1}
Suppose that $\A \in T(\HH)$ is right-side symmetric.   Then $\A$ has three L-eigenvalues $\sigma_1 \ge \sigma_2 \ge \sigma_3 \ge 0$, with three symmetric L-eigentensors $\V_1, \V_2, \V_3$, and L-eigenvectors $\x_1, \x_2, \x_3$, associated with $\sigma_1, \sigma_2$ and $\sigma_3$ correspondingly.
\end{Theorem}
\begin{proof} If $\A$ is nonsingular, then all of its three L-eigenvalues are positive.  The conclusion follows from Conclusion 6 of Theorem \ref{t3}.   Suppose that $\A$ is singular, with L-eigenvalues $\sigma_1 \ge \sigma_2 \sigma_3 = 0$.   We may construct sequences $\{ \A_k \in T(\HH) : k = 1, 2, \cdots \}$ such that $\A_k$ is nonsingular for $k = 1, 2, \cdots$, and the representative hypermatrix of  $\A_k$ converges to the representative hypermatrix of  $\A$, under an orthonormal basis.   Suppose that $\A_k$ has three L-eigenvalues $\sigma_{1, k} \ge \sigma_{2, k} \ge \sigma_{3, k} > 0$, with associated L-eigenvector $\x_{1, k}, \x_{2, k}, \x_{3, k}$ and symmetric L-eigentensor $\V_{1, k}, \V_{2, k}, \V_{3, k}$ for $k = 1, 2, \cdots$.
Singular values of matrices are roots of eigenvalues of square matrices, which are continuous with respect to the entries of such square matrices.  Thus singular values of matrices are also continuous with respect to the entries of such matrices.  This implies that $\sigma_{i, k}$ converges to $\sigma_i$ as $k$ tends to infinity, for $i  = 1, 2, 3$.   Since $\x_{i, k} \bullet \x_{i, k} = 1$ and $\V_{i, k} \bullet \V_{i, k} = 1$ for $i = 1, 2, 3$ and $k = 1, 2, \cdots$, their representative vectors and representative matrices have limiting points for $i = 1, 2, 3$. Taking such limiting points, we have $\x_1, \x_2, \x_3 \in \HH$ and $\V_1, \V_2, \V_3 \in B(\HH)$ such that they are L-eigenvectors and L-eigentensors of $\A$, associated with L-eigenvalues $\sigma_1, \sigma_2, \sigma_3$, respectively.   However, such $\V_1, \V_2, \V_3$ should be symmetric as their representative matrices are limiting points of symmetric matrices.
\end{proof}

Suppose that $\V_1, \V_2, \V_3$ are such symmetric L-eigentensors.  Let the eigenvalues of $\V_i$ be $\lambda_{i, 1}, \lambda_{i, 2}, \lambda_{i, 3}$ with associated eigenvectors $\y_{i, 1}, \y_{i, 2}, \y_{i, 3}$ such that $\y_{i, j} \bullet \y_{i, k} = \delta_{jk}$ for $i, j, k = 1, 2, 3$.   By (\ref{decomp}), we have
$$\V_j = \sum_{k=1}^3 \lambda_{j, k} \y_{j, k} \otimes \y_{j, k},$$
for $j = 1, 2, 3$.  Substituting them to (\ref{e4}), we have
\begin{equation} \label{third-decomp}
\A = \sum_{j, k=1}^3 \sigma_j \lambda_{j, k} \x_j \otimes \y_{j, k} \otimes \y_{j, k}.
\end{equation}
Here, $\sigma_1 \ge \sigma_2 \ge \sigma_3$, $x_j \bullet \x_k = \delta_{jk}, \y_{i, j} \bullet \y_{i, k} = \delta_{jk}$ for $i, j, k = 1, 2, 3$.

Suppose that $\A \in T(\HH)$ is left-side symmetric.   Then $(\A^\top)^\top$ is right-side symmetric.   Then, we have
$$(\A^\top)^\top = \sum_{j, k=1}^3 \bar \sigma_j \bar \lambda_{j, k} \bar \x_j \otimes \bar \y_{j, k} \otimes \bar \y_{j, k},$$
where $\bar \sigma_1 \ge \bar \sigma_2 \ge \bar \sigma_3$, $\bar x_j \bullet \bar \x_k = \delta_{jk}, \bar \y_{i, j} \bullet \bar \y_{i, k} = \delta_{jk}$ for $i, j, k = 1, 2, 3$.
By Proposition \ref{p1}, we have
\begin{equation} \label{third-decomp-1}
\A = \sum_{j, k=1}^3 \bar \sigma_j \bar \lambda_{j, k}  \bar \y_{j, k} \otimes \bar \y_{j, k} \otimes \bar \x_j.
\end{equation}

Similarly, for a third order centrally symmetric tensor $\A$, we have
\begin{equation} \label{third-decomp-2}
\A = \sum_{j, k=1}^3 \hat \sigma_j \hat \lambda_{j, k}  \hat \y_{j, k} \otimes \hat \x_j \otimes \hat \y_{j, k},
\end{equation}
where $\hat \sigma_1 \ge \hat \sigma_2 \ge \hat \sigma_3$, $\hat x_j \bullet \hat \x_k = \delta_{jk}, \hat \y_{i, j} \bullet \hat \y_{i, k} = \delta_{jk}$ for $i, j, k = 1, 2, 3$.

Thus, a third order symmetric tensor $\A$ has all of these three eigenvector decomposition.

\medskip

Let $\lambda_i \in \Re, \x_i \in \HH$ for $i = 1, 2, 3$ such that $\x_i \bullet \x_j = \delta_{ij}$ for $i, j = 1, 2, 3$.  Let
\begin{equation} \label{sym-decomp}
\A = \sum_{i=1}^3 \lambda_i  \x_i \otimes \x_i \otimes \x_i.
\end{equation}
Then $\A$ is third order symmetric tensor.   We call such a symmetric tensor $\A$ a third order primarily symmetric tensor.   Can any third order symmetric tensor be decomposed as the sum of such third order primarily symmetric tensors?

\medskip

Let $\lambda_i \in \Re, \x_i \in \HH$ for $i = 1, 2, 3$ such that $\x_i \bullet \x_j = \delta_{ij}$ for $i, j = 1, 2, 3$.  Let
\begin{equation} \label{cyc-sym-decomp}
\A =  \lambda_1  \x_1 \otimes \x_2 \otimes \x_3 + \lambda_2 \x_2 \otimes \x_3 \otimes \x_1 + \lambda_3  \x_3 \otimes \x_1 \otimes \x_2.
\end{equation}
Then $\A$ is third order cyclically symmetric tensor.   We call such a symmetric tensor $\A$ a third order primarily cyclically symmetric tensor.   Can any third order cyclically symmetric tensor be decomposed as the sum of such third order primarily cyclically symmetric tensors?

\bigskip

In the discussion of eigenvector decompositions (\ref{third-decomp}-\ref{third-decomp-2}), we see that for a third order tensor $\A$, its kernel tensor $\U$, the kernel tensor $\bar \U$ of $\A^\top$, and the kernel tensor $\hat \U$ of $(\A^\top)^\top$ play important roles.   The invariants tr$(\U)$, tr$(\U^2)$, tr$(\U^3)$ of $\U$, the invariants tr$(\bar \U)$, tr$(\bar \U^2)$, tr$(\bar \U^3)$ of $\bar \U$, and the invariants tr$(\hat \U)$, tr$(\hat \U^2)$, tr$(\hat \U^3)$ of $\hat \U$ are invariants of $\A$.   They are polynomials of entries of the representative matrices of $\U, \bar \U, \hat \U$, thus polynomials of entries of the representative hypermatrix $\A$.  They are easily computable.    Suppose the representative hypermatrix of $\A$ is $a_{ijk}$.  Then
$${\rm tr}(\U) = {\rm tr}(a_{ijk}a_{ljk}) = a_{ijk}a_{ijk} = \A \bullet \A.$$
Similarly, we can show that ${\rm tr}(\bar \U) = \A \bullet \A$ and ${\rm tr}(\hat \U) = \A \bullet \A$.
In this way, we have seven invariants tr$(\U)$, tr$(\U^2)$, tr$(\U^3)$, tr$(\bar \U^2)$, tr$(\bar \U^3)$, tr$(\hat \U^2)$, tr$(\hat \U^3)$  of $\A$.   Surely these seven invaiants are not complete in the sense of \cite{Zh-94}.   If $\A$ is cyclically symmetric, then we have $\U = \bar \U = \hat \U$, and hence tr$(\U^2) =$ tr$(\bar \U^2) =$ tr$(\hat \U^2)$ and tr$(\U^3) =$ tr$(\bar \U^3) =$ tr$(\hat \U^3)$.  In general, do we have some relations among these seven invariants?

\bigskip

A first order $n$-dimensional tensor only has one independent invariant, though its representative vector has $n$ independent components.   A second order tensor $\U \in B(\HH)$ has only six independent invariants, though its representative matrix has nine independent entries.  If $\U$ is symmetric, then it has only three independent invariants, though its representative symmetric matrix has six independent entries.   See \cite{HXL-04, Zh-94}.

The representative hypermatrix of a third order $\A \in T(\HH)$ has 27 independent entries.   If $\A$ is partially symmetric, or cyclically symmetric, or symmetric,  then its representative hypermatrix has 18, or 11, or 10 independent entries, respectively.  These are easy to be seen.   Then how many independent invariants do such third order tensors have?

We make the following conjecture and will explain the reason in the next section.

\begin{Conjecture}  \label{conj-2}
If $\A \in T(\HH)$ is symmetric, then it has seven independent invariants.
\end{Conjecture}

Then we raise the following question.

\begin{Question}
If $\A \in T(\HH)$, then how many independent invariants does it have?  If it is partially symmetric, then how many independent invariants does it have? If it is cyclically symmetric, then how many independent invariants does it have?
\end{Question}

\section{Singular Values, C-Eigenvalues and Z-Eigenvalues}

We now define singular values of a third order tensor.

Let $\A \in T(\HH),  \x, \y, \z  \in \HH$ and $\eta \in \Re$.    We say that $\eta$ is a singular value of $\A$, $\x, \y$ and $\z$ are associated left, central and right singular vectors respectively if $\eta \ge 0$ and the following equations hold.
\begin{equation} \label{e7}
\A\y\z = \eta \x, \ \ \ \x\A\z = \eta \y, \ \ \ \x\y\A = \eta\z, \ \ \ \x \bullet \x = 1, \ \ \  \y \bullet \y = 1, \ \ \  \z \bullet \z = 1.
\end{equation}

\begin{Theorem} \label{t4}   Suppose that $\A \in T(\HH)$.   Then $\A$ always has singular values.  They are invariants of $\A$.    Let $\eta$ be a singular value of $\A$, with associated left, central and right singular vectors $\x, \y$ and $\z$.    Then $\eta = \x\A\y\z$.   Let $\eta_1$ be the maximum singular value of $\A$.    Then
\begin{equation} \label{e8}
\eta_1 = \max \{ \x\A\y\z : \x \bullet \x = 1, \y \bullet \y = 1, \z \bullet \z = 1 \}.
\end{equation}
\end{Theorem}
\begin{proof}     Suppose that $\A, \x, \y$ and $\z$ are represented by $a_{ijk}, x_i, y_j$ and $z_k$ under an orthogonal basis.    Then $\x\A\y\z = a_{ijk}x_iy_jz_k$, $x_ix_i = 1, y_jy_j = 1 $ and $z_kz_k = 1$, as these values are invariant with respect to orthonormal transformation.   We have
\begin{equation} \label{e9}
\max \{ \x\A\y\z : \x \bullet \x = 1, \y \bullet \y = 1, \z \bullet \z = 1 \} = \max \{ a_{ijk}x_iy_jz_k : x_ix_i = 1, y_jy_j = 1, z_kz_k = 1 \}.
\end{equation}
Consider the optimality conditions of $\max \{ a_{ijk}x_iy_jz_k : x_ix_i = 1, y_jy_j = 1, z_kz_k = 1 \}$.   We have optimizers $x_i, y_j, z_k$, Lagrangian multipliers $\eta, \eta'$ and $\eta''$ such that
$$a_{ijk}y_jz_k = \eta x_i, \ \ \ a_{ijk}x_iz_k = \eta'y_j, \ \ \ a_{ijk}x_iy_j = \eta''z_k, \ \ \  x_ix_i = 1, \ \ \ y_jy_j = 1, \ \ \ z_kz_k = 1.$$
Multiplying the first equality by $x_i$, the second equality by $y_j$, the third equality by $z_k$, we have
$$\eta = \eta' = \eta'' = a_{ijk}x_iy_jz_k.$$
Thus, we have
$$a_{ijk}y_jz_k = \eta x_i, \ \ \ a_{ijk}x_iz_k = \eta y_j, \ \ \ a_{ijk}x_iy_j = \eta z_k, \ \ \  x_ix_i = 1, \ \ \ y_jy_j = 1, \ \ \ z_kz_k = 1.$$
This implies (\ref{e7}), and shows the existence of singular values, and also $\eta = \x\A\y\z$ if $\eta$ is  a singular value of $\A$, with associated left, central and right singular vector $\x, \y$ and $\z$.   By (\ref{e9}), we have (\ref{e8}).
\end{proof}

We may regard $\x\A\y\z$ as a potential of $\A$.   Then (\ref{e8}) shows that the largest singular value of $\A$ gives the largest value of this potential.

\medskip

C-eigenvalues were introduced in \cite{CJQ-17} for third order right-side symmetric tensors. Also see \cite{LL-17}.  We may extend them to general third order tensors.   Here, ``C'' was named after Curier brothers.

Let $\A \in T(\HH),  \x, \y,  \in \HH$ and $\mu \in \Re$.    We say that $\mu$ is a C-eigenvalue of $\A$, $\x$ and $\y$ are associated left and right C-eigenvectors respectively if $\mu \ge 0$ and the following equations hold.
\begin{equation} \label{e10}
\A\y\y = \mu \x, \ \ \ \x\A\y = \mu \y, \ \ \ \x \bullet \x = 1, \ \ \  \y \bullet \y = 1.
\end{equation}
Similarly, we can prove the following theorem.

\begin{Theorem} \label{t5}   Suppose that $\A \in T(\HH)$ is right-side symmetric.   Then $\A$ always has C-eigenvalues.  They are invariants of $\A$.    Let $\mu$ be a C-eigenvalue of $\A$, with associated left and right C-eigenvectors $\x$ and $\y$.    Then $\mu = \x\A\y\y$.   Let $\mu_1$ be the maximum C-eigenvalue of $\A$.    Then
\begin{equation} \label{e11}
\mu_1 = \max \{ \x\A\y\y : \x \bullet \x = 1, \y \bullet \y = 1 \}.
\end{equation}
Comparing (\ref{e11}) with (\ref{e8}), for a third-order right-side symmetric tensor $\A$, we have $\mu_1 \le \eta_1$.
\end{Theorem}

In \cite{CJQ-17}, it was pointed out that if $\A$ is the piezoelectric tensor, then $\mu_1$ is the highest piezoelectric coupling constant.

\medskip

Z-eigenvalues were introduced in \cite{Qi-05} for symmetric tensors.   They were extended to nonsymmetric tensors in \cite{Qi-07}. We now discuss them in the context of third order tensors.

Let $\A \in T(\HH),  \x,  \in \HH$ and $\nu \in \Re$.    We say that $\nu$ is a Z-eigenvalue of $\A$, $\x$ is an associated Z-eigenvectors respectively if $\nu \ge 0$ and the following equations hold.
\begin{equation} \label{e12}
\A\x\x = \nu \x, \ \ \ \x \bullet \x = 1.
\end{equation}
Similarly, we can prove the following theorem.

\begin{Theorem} \label{t6}   Suppose that $\A \in T(\HH)$ is symmetric.   Then $\A$ always has Z-eigenvalues.  They are invariants of $\A$.    Let $\nu$ be a Z-eigenvalue of $\A$, with associated Z-eigenvector $\x$.    Then $\nu = \x\A\x\x$.   Let $\nu_1$ be the maximum Z-eigenvalue of $\A$.    Then
\begin{equation} \label{e13}
\nu_1 = \max \{ \x\A\x\x : \x \bullet \x = 1 \}.
\end{equation}
Comparing (\ref{e13}) with (\ref{e8}) and (\ref{e11}), for a third-order symmetric tensor $\A$, we have $\nu_1 \le \mu_1 \le \eta_1$.  If $\A$ is symmetric, then $\nu_1 = \mu_1$.
\end{Theorem}
\begin{proof}   The other conclusions can be proved similarly as the proof of Theorem \ref{t4}.   The last conclusion that if $\A$ is symmetric, then $\nu_1 = \mu_1$, follows from Theorem 2.2 of \cite{ZQY-12}.
\end{proof}

As the proof of the above theorem shows, when $\A$ is symmetric, the content of this section is related with the discussion on cubic spherical optimization problems in \cite{ZQY-12}, where some computational method also can be used for computing $\nu_1$, $\mu_1$ and $\eta_1$ here.

By \cite{Qi-05}, the complex version of Z-eigenvalues are called E-eigenvalues.   By \cite{CS-13}, a third order three dimensional generic symmetric hypermatrix has seven E-eigenvalues.   Thus, we made the conjecture that a third order symmetric tensor $\A \in T(\HH)$ has seven independent invariants.

\section{The Levi-Civita Tensor and Related Third Order Tensors}

We now study more about the Levi-Civita tensor $\E$.    We have the following theorem.

\begin{Theorem} \label{t7}  The Levi-Civita Tensor $\E$ is nonsingular.   Its kernel tensor is $2\I$, $\E^{-1} = {1 \over 2}\E$, and its three L-eigenvalues are $\sigma_1 = \sigma_2 = \sigma_3 = \sqrt{2}$.   Let $\z \in \HH$. If $\U = \E\z$, then $\z = {1 \over 2}\E\U$.   Furthermore, ${1 \over \sqrt{2}}\E$ is a third order orthogonal tensor.
\end{Theorem}
\begin{proof}   We now that $\E$ is cyclically symmetric, $\E^\T = \E$.   By direct calculation, we have $\epsilon_{ijk}\epsilon_{jkl} = 2\delta_{il}$.  Thus, $\E\E^\T = 2\I$. By Theorem \ref{t2}, $\E$ is nonsingular, and $\E^{-1} = {1 \over 2}\E$. By Proposition \ref{p1}, if $\U = \E\z$, then $\z = {1 \over 2}\E\U$.  Since $\E\E^\T = 2\I$, by Theorem \ref{t3}, the L-eigenvalues of $\E$ are $\sigma_1 = \sigma_2 = \sigma_3 = \sqrt{2}$.    Since ${1 \over \sqrt{2}}\E({1 \over \sqrt{2}}\E)^\T = \I$, by definition, ${1 \over \sqrt{2}}\E$ is a third order orthogonal tensor.
\end{proof}

Now, a question is: what is the largest singular value of $\E$?

\medskip

In Section 3, we see that right-side symmetry and left-side symmetry are tensor properties.    We may now define some tensor symmetry properties by using the Levi-Civita tensor.      Let $a_{ijk}$ be a hypermatrix.   We say that the hypermatrix $a_{ijk}$ is selectively right-side symmetric if $a_{ijk} = a_{ikj}$ for $k \not = i \not = j$.    Similarly, we say that the hypermatrix $a_{ijk}$ is selectively left-side symmetric if $a_{ijk} = a_{jik}$ for $i \not = k \not = j$.

\begin{Proposition} \label{p4}
Let $\A \in T(\HH)$.   Then the representative hypermatrix of $\A$ under an orthonormal basis is selectively right-side symmetric if and only if $\A\E = \OO$, and the representative hypermatrix of $\A$ under an orthonormal basis is selectively left-side symmetric if and only if $\E\A = \OO$.  Thus, the selectively right-side symmetric property and the selectively left-side symmetric property are tensor properties.
\end{Proposition}
\begin{proof}  By direct calculation, we see that a hypermatrix $a_{ijk}$ is selectively right-side symmetric if and only if $a_{ijk}\epsilon_{jkl} = 0_{il}$, $a_{ijk}$ is selectively left-side symmetric if and only if $\epsilon_{ijk}a_{jkl} = 0_{il}$.  The conclusion follows.
\end{proof}

Let $a_{ijk}$ be a hypermatrix.   Suppose that $a_{iik} = a_{iki}$ for all $i$ and $k$.  Here, double $i$ does not mean summation.   Is such a property a tensor property?

\section{Final Remarks}

In this paper, we studied third order tensors by regarding them as linear operators transforming first order tensors into second order tensors, or second order tensors into first order tensors.
From this point of view, various new concepts have been introduced, and various results and invariants have been derived.    Some questions have been raised.   We hope that we may obtain more useful results for third order tensors from this approach.

\bigskip

{\bf Acknowledgment.}  The author is thankful to Quanshui Zheng, Wennan Zou, Yannan Chen, Weiyang Ding, Diwei Shi and Jinjie Liu for the discussion.

\end{document}